\documentclass[journal]{IEEEtran}%
\usepackage{amsfonts}
\usepackage{amsmath}
\usepackage{amssymb}
\usepackage{graphicx}
\usepackage{xcolor}
\usepackage{pgf}
\usepackage{cite}%
\newtheorem{theorem}{Theorem}

\newtheorem{definition}[theorem]{Definition}
\newtheorem{example}[theorem]{Example}

\newtheorem{lemma}[theorem]{Lemma}

\newtheorem{proposition}[theorem]{Proposition}

\begin{document}

\title{On Pairs of $f$-divergences and their Joint Range}
\author{Peter Harremo\"{e}s,~\IEEEmembership{Member,~IEEE,} Igor Vajda\dag
,~\IEEEmembership{Fellow ~IEEE} \thanks{Manuscript received xxxxx, 2010;
revised xxxxx. This work was supported by the European Network of Excellence
and the GA\v{C}R grants 102/07/1131 and 202/10/0618.}\thanks{Peter
Harremo\"{e}s is with Copenhagen business College, Denmark. Igor Vajda passed
away during the preparation of this article. He worked at Institute of
Information and Automation, Prague, Czech Republic.}}
\pubid{0000--0000/00\$00.00~\copyright ~2010 IEEE}
\maketitle

\begin{abstract}
We compare two $f$-divergences and prove that their joint range is the convex
hull of the joint range for distributions supported on only two points. Some
applications of this result are given.

\end{abstract}


\section{Divergences and divergence statistics}

\PARstart{M}{any} of the divergence measures used in statistics are of the
$f$-divergence type introduced independently by I. Csisz\'{a}r
\cite{Csiszar1963}, T. Morimoto \cite{morimoto1963}, and Ali and Silvey
\cite{Ali1966}. Such divergence measures have been studied in great detail in
\cite{Liese1987}. Often one is interested inequalities for one $f$-divergence
in terms of another $f$-divergence. Such inequalities are for instance needed
in order to calculate the relative efficiency of two $f$-divergences when used
for testing goodness of fit but there are many other applications. In this
paper we shall study the more general problem of determining the joint range
of any pair of $f$-divergences. The results are useful in determining general
conditions under which information divergence is a more efficient statistic
for testing goodness of fit than another $f$-divergence, but will not be
discussed in this short paper.

Let $f:\left(  0,\infty\right)  \rightarrow\mathbb{R}$ denote a convex
function satisfying $f\left(  1\right)  =0.$ We define $f\left(  0\right)  $
as the limit $\lim_{t\rightarrow0}f\left(  t\right)  $. We define $f^{\ast
}\left(  t\right)  =tf\left(  t^{-1}\right)  .$ Then $f^{\ast}$ is a convex
function and $f^{\ast}\left(  0\right)  $ is defined as $\lim_{t\rightarrow
0}tf\left(  t^{-1}\right)  =\lim_{t\rightarrow\infty}\frac{f\left(  t\right)
}{t}.$ Assume that $P$ and $Q$ are absolutely continuous with respect to a
measure $\mu,$ and that $p=\frac{dP}{d\mu}$ and $q=\frac{dQ}{d\mu}.$ For
arbitrary distributions $P$ and $Q$ the $f$-divergence $D_{f}(P,Q)\geq0$\ is
defined by the formula
\begin{equation}
D_{f}(P,Q)=\int_{\left\{  q>0\right\}  }f\left(  \frac{p}{q}\right)
~dQ+f^{\ast}\left(  0\right)  P\left(  q=0\right)  \label{4}%
\end{equation}
(for details about the definition (\ref{4}) and properties of the
$f$-divergences, see \cite{Liese2006}, \cite{Liese1987} or \cite{Read1988}).
With this definition
\[
D_{f}\left(  P,Q\right)  =D_{f^{\ast}}\left(  Q,P\right)  .
\]

\begin{example}
The function $f(t)=\left\vert t-1\right\vert $ defines the $L^{1}$-distance%
\begin{equation}
\left\Vert P-Q\right\Vert =\sum_{j=1}^{k}q_{j}\,\left\vert {\frac{p_{j}}%
{q_{j}}-1}\right\vert =\sum_{j=1}^{k}\,\left\vert p_{j}{-q_{j}}\right\vert
\text{ \ \ (cf. (\ref{4}))} \label{V}%
\end{equation}
which plays an important role in information theory and mathematical
statistics \cite{Barron1992, Fedotov2003} . \input{pinsker.TpX}
\end{example}

In (\ref{4}) is often taken the convex function $f$ which is one of the power
functions $\phi_{\alpha}$\ of order $\alpha\in\mathbb{R}$ given in the domain
$t>0$ by the formula
\begin{equation}
\phi_{\alpha}(t)={\frac{t^{\alpha}-\alpha(t-1)-1}{\alpha(\alpha-1)}}\text{
\ \ \ when \ }\alpha(\alpha-1)\neq0 \label{4a}%
\end{equation}
and by the corresponding limits
\begin{equation}
\phi_{0}(t)=-\ln t+t-1\text{ \ \ and \ \ }\phi_{1}(t)=t\ln t-t+1. \label{4b}%
\end{equation}
The $\phi$-divergences
\begin{equation}
D_{\alpha}(P,Q)\overset{def}{=}D_{\phi_{\alpha}}(P,Q),\text{ \ \ }\alpha
\in\mathbb{R} \label{4c}%
\end{equation}
based on (\ref{4a}) and (\ref{4b}) are usually referred to as power
divergences of orders $\alpha.$ For details about the properties of power
divergences, see \cite{Liese2006} or \cite{Read1988}. Next we\ mention the
best known members of the family of statistics (\ref{4c}), with a reference to
the skew symmetry $D_{\alpha}(P,Q)=D_{1-\alpha}(Q,P)$ of the power divergences
(\ref{4c}).$\medskip$

\begin{example}
The $\chi^{2}$-divergence (or quadratic divergence or Pearson divergence)
\begin{equation}
D_{2}(P,Q)=D_{-1}(Q,P)={\frac{1}{2}}\sum_{j=1}^{k}{\frac{(p_{j}-q_{j})^{2}%
}{q_{j}}} \label{chi}%
\end{equation}
leads to the well known Pearson and Neyman statistics. The information
divergence
\begin{equation}
D_{1}(P,Q)=D_{0}(Q,P)=\sum_{j=1}^{k}p_{j}\ln{\frac{p_{j}}{q_{j}}} \label{7}%
\end{equation}
leads to the log-likelihood ratio and reversed log-likelihood ratio
statistics. The symmetric Hellinger divergence
\[
D_{1/2}(P,Q)=D_{1/2}(Q,P)=H(P,Q)
\]
leads to the Freeman--Tukey statistic.
\end{example}

\begin{example}
The Hellinger divergence and the total variation are symmetric in the
arguments $P$ and $Q.$ Non-symmetric divergences may be symmetrized. For
instance the LeCam divergence is nothing but the symmetrized $\chi^{2}%
$-divergence given by
\[
D_{LeCam}\left(  P,Q\right)  =\frac{1}{2}D_{2}\left(  P,\frac{P+Q}{2}\right)
+\frac{1}{2}D_{2}\left(  Q,\frac{P+Q}{2}\right)
\]
Another symmetrized divergence is the Jensen Shannon divergence defined by
\[
JD_{1}\left(  P,Q\right)  =\frac{1}{2}D\left(  P\left\Vert \frac{P+Q}%
{2}\right.  \right)  +\frac{1}{2}D\left(  Q\left\Vert \frac{P+Q}{2}\right.
\right)  .
\]
The joint range of total variation with Jensen Shannon divergence was studied
by Bri\"{e}t and Harremo\"{e}s \cite{Briet2009} and is illustrated on Figure
\ref{vsjd}.

\input{vsjd.TpX}
\end{example}

In this paper we shall prove that the joint range of any pair of
$f$-divergences is essentially determined by the range of distributions on a
two-element set. In special cases the significance of determining the range
over two-element set has been pointed out explicitly in \cite{Topsoe2001a}.
Here we shall prove that a reduction to two-element sets can always be made.

\section{\label{sec1}Joint range of $f$-divergences}

In this section we are interested in the range of the map $\left(  P,Q\right)
\rightarrow\left(  D_{f}\left(  P,Q\right)  ,D_{g}\left(  P,Q\right)  \right)
$ where $P$ and $Q$ are probability distributions on the same set.

\begin{definition}
A point $\left(  x,y\right)  \in\mathbb{R}^{2}$ is $(f,g)$\emph{-achievable}
if there exist probability measures $P$ and $Q$ on a $\sigma$-algebra such
$\left(  x,y\right)  =\left(  D_{f}\left(  P,Q\right)  ,D_{g}\left(
P,Q\right)  \right)  .$ A $(f,g)$-divergence pair $\left(  x,y\right)  $ is
$d$-\emph{achievable } if there exist probability vectors $P,Q\in
\mathbb{R}^{d}$ such that
\[
\left(  x,y\right)  =\left(  D_{f}\left(  P,Q\right)  ,D_{g}\left(
P,Q\right)  \right)  .
\]

\end{definition}

\begin{lemma}
Assume that
\[
P_{0}\left(  A\right)  =Q_{0}\left(  A\right)  =1
\]
and
\[
P_{1}\left(  B\right)  =Q_{1}\left(  B\right)  =1
\]
and that $A\cap B=\varnothing.$ If $P_{\alpha}=\left(  1-\alpha\right)
P_{0}+\alpha P_{1}$ and $Q_{\alpha}=\left(  1-\alpha\right)  Q_{0}+\alpha
Q_{1} $ then
\[
D_{f}\left(  P_{\alpha},Q_{\alpha}\right)  =\left(  1-\alpha\right)
D_{f}\left(  P_{0},Q_{0}\right)  +\alpha D_{f}\left(  P_{1},Q_{1}\right)  .
\]

\end{lemma}

\begin{theorem}
\label{TheoremConvex}The set of $(f,g)$-achievable points is convex.
\end{theorem}

\begin{proof}
Assume that $\left(  P,Q\right)  $ and $\left(  \tilde{P},\tilde{Q}\right)  $
are two pairs of probability distributions on a space $\left(  \mathcal{X}
,\mathcal{F}\right)  .$ Introduce a two-element set $B=\left\{  0,1\right\}  $
and the product space $\mathcal{X\times}B$ as a measurable space. Let $\phi$
denote projection on $B.$ Now we define a pair $\left(  \tilde{P},\tilde
{Q}\right)  $of joint distribution on $\mathcal{X\times}B.$ The marginal
distribution of both $\tilde{P}$ is $\tilde{Q}$ on $B$ is $\left(
1-\alpha,\alpha\right)  .$ The conditional distributions are given by
$P\left(  \cdot\mid\phi=i\right)  =P_{i}$ and $Q\left(  \cdot\mid
\phi=i\right)  =Q_{i}$ where $i=0,1.$ Then
\begin{multline*}
\left(
\begin{array}
[c]{c}%
D_{f}\left(  P_{\alpha},Q_{\alpha}\right) \\
D_{g}\left(  P_{\alpha},Q_{\alpha}\right)
\end{array}
\right)  =\\
\left(
\begin{array}
[c]{c}%
\left(  1-\alpha\right)  D_{f}\left(  P_{0},Q_{0}\right)  +\alpha D_{f}\left(
P_{1},Q_{1}\right) \\
\left(  1-\alpha\right)  D_{g}\left(  P_{0},Q_{0}\right)  +\alpha D_{g}\left(
P_{1},Q_{1}\right)
\end{array}
\right) \\
=\left(  1-\alpha\right)  \left(
\begin{array}
[c]{c}%
D_{f}\left(  P_{0},Q_{0}\right) \\
D_{g}\left(  P_{0},Q_{0}\right)
\end{array}
\right)  +\alpha\left(
\begin{array}
[c]{c}%
D_{f}\left(  P_{1},Q_{1}\right) \\
D_{g}\left(  P_{1},Q_{1}\right)
\end{array}
\right) \\
=\left(  1-\alpha\right)  \left(
\begin{array}
[c]{c}%
D_{f}\left(  P,Q\right) \\
D_{g}\left(  P,Q\right)
\end{array}
\right)  +\alpha\left(
\begin{array}
[c]{c}%
D_{f}\left(  \tilde{P},\tilde{Q}\right) \\
D_{g}\left(  \tilde{P},\tilde{Q}\right)
\end{array}
\right)  .
\end{multline*}

\end{proof}

\begin{example}
For the joint range of total variation and Jensen Shannon divergence
illustrated on Figure \ref{vsjd} the set of 2-achievable points is not convex
but the set of 3-achievable points is convex and equals the set of all
$(f,g)$-achievable points.
\end{example}

\begin{theorem}
Any $(f,g)$-achievable points is a convex combination of two $2$-achievable
points. Consequently, any $(f,g)$-achievable point is $4$-achievable.
\end{theorem}

\begin{proof}
Let $P$ and $Q$ denote probability measures on Borel space. Define the set
$A=\left\{  q>0\right\}  $ and the function $X=p/q$ on $A.$ Then $Q$
satisfies
\begin{align}
Q\left(  A\right)   &  =1,\label{norm}\\
\int_{A}X~dQ  &  \leq1.\nonumber
\end{align}
Now we fix $X$ and $A.$ The formulas for the divergences become
\begin{align*}
D_{f}\left(  P,Q\right)   &  =\int_{A}f\left(  X\right)  ~dQ+f^{\ast}\left(
0\right)  P\left(  \complement A\right) \\
&  =\int_{A}f\left(  X\right)  ~dQ+f^{\ast}\left(  0\right)  \left(
1-\int_{A}X~dQ\right) \\
&  =\int_{A}\left(  f\left(  X\right)  ~+f^{\ast}\left(  0\right)  \left(
1-X\right)  \right)  ~dQ\\
&  =\mathrm{E}\left[  f\left(  X\right)  +f^{\ast}\left(  0\right)  \left(
1-X\right)  \right]
\end{align*}
and similarly
\[
D_{g}\left(  P,Q\right)  =\mathrm{E}\left[  g\left(  X\right)  ~+g^{\ast
}\left(  0\right)  \left(  1-X\right)  \right]  .
\]
Hence, the divergences only depend on the distribution of $X.$ Therefore we
may without loss of generality assume that $Q$ is a probability measure on
$\left[  0,\infty\right)  $.

Define $C$ as the set of probability measures on $\left[  0,\infty\right)  $
satisfying $\mathrm{E}\left[  X\right]  \leq1.$ Let $C^{+}$ be the set of
additive measures $\mu$ on $\left[  0,\infty\right)  $ satisfying $\mu\left(
A\right)  \leq1$ and $\int_{A}X~d\mu\leq1.$ Then $C^{+}$ is convex and thus
compact under setwise convergence. According to the Choquet--Bishop--de Leeuw
theorem \cite[Sec. 4]{Phelps2001} any other point in $C^{+}$ is the barycenter
of a probability measure over such extreme points. In particular an element
$Q\in C$ is the barycenter of a probability measure $P_{bary}$ over extreme
points of $C^{+}$ and these extreme points must in addition be probability
measures with $P_{bary}$-probability 1. Hence $Q\in C$ is a barycenter of a
probability measure over extreme points in $C.$

Let $Q$ be an element in $C.$ Let $A_{i},i=1,2,3$ be a disjoint cover of
$\left[  0,\infty\right)  $ and assume that $Q\left(  A_{i}\right)  >0.$ Then
\[
Q=\sum_{i=1}^{3}Q\left(  A_{i}\right)  Q\left(  \cdot\mid A_{i}\right)  .
\]
For a probability vector $\lambda=\left(  \lambda_{1},\lambda_{2},\lambda
_{2}\right)  $ let $Q_{\lambda}$ denote the distribution
\[
Q_{\lambda}=\sum_{i=1}^{3}\lambda_{i}Q\left(  \cdot\mid A_{i}\right)  .
\]
Then $Q_{\lambda}$ is element in $C$ if and only if
\begin{equation}
\sum_{i=1}^{3}\lambda_{i}\int_{A}X~dQ\left(  \cdot\mid A_{i}\right)  \leq1.
\label{reduceret}%
\end{equation}
An extreme probability vector $\lambda$ that satisfies (\ref{reduceret}) has
one or two of its weights equal to 0. Hence, if $Q$ is extreme in $C$ and
$A_{i},i=1,2,3$ is a disjoint cover of $A,$ then at least one of the three
sets satisfies $Q\left(  A_{i}\right)  =0.$ Therefore an extreme point $Q\in
C$ is of one of the following two types:

\begin{enumerate}
\item $Q$ is concentrated in one point.

\item $Q$ has support on two points. In this case the inequality $\int
_{A}X~dQ\leq1$ holds with equality and $P\left(  A\right)  =1$ so that $P$ is
absolutely continuous with respect to $Q$ and therefore supported by the same
two-element set.
\end{enumerate}

The formulas for divergence are linear in $Q.$ Hence any $(f,g)$-divergence
pair is a the barycenter of a probability measure $P_{bary}$ over points
generated by extreme distributions $Q\in C.$ The extreme distributions of type
$2$ generate 2-achievable points.

For extreme points $Q$ concentrated in a single point we can reverse the
argument at make a barycentric decomposition with respect to $P$. If an
extreme $P$ has a two-point support then $Q$ is absolutely continuous with
respect to $P$ and generates a $(f,g)$-achievable point that is $2$%
-achievable. If $P$ is concentrated in a point then this point may either be
identical with the support of $Q$ and the two probability measures are
identical, or the support points are different and $P$ and $Q$ are singular
but still $\left(  P,Q\right)  $ is supported on two points. Therefore any
$(f,g)$-achievable point has a barycentric decomposition into 2-achievable points.

\input{trekant.TpX}

Let $\mathbf{y}=\left(  y,z\right)  $ be a $(f,g)$-achievable point. As we
have seen $\mathbf{y}$ is a barycenter of $(f,g)$-achievable points that are
2-achievable. According to the Carath\'{e}odory's theorem
\cite{Boltyanski2001} any barycentric decomposition in two dimensions may be
obtained as a convex combination of at most three points $\mathbf{y}%
_{i},~i=1,2,3.$ as illustrated in Figure \ref{trekant}. Assume that all three
points have positive weight. Let $\ell_{i}$ be the line through $\mathbf{y}$
and $\mathbf{y}_{i}.$ The point $\mathbf{y}$ divides the line $\ell_{i}$ in
two half-lines $\ell_{i}^{+}$ and $\ell_{i}^{-}~,$ where $\ell_{i}^{-}$
denotes the half-line that contains $\mathbf{y}_{i}.$ The lines $\ell_{i}%
^{+},i=1,2,3$ divide $\mathbb{R}^{2}$ into three sectors, each of them
containing one of the points $\mathbf{y}_{i},i=1,2,3.$ The set of
$(f,g)$-divergence pairs that are $3$-achievable is curve-connected so there
exist a continuous curve of $(f,g)$-divergence pairs that are 2-achievable
from $\mathbf{y}_{1}$ to $\mathbf{y}_{2}$ that must intersect $\ell_{1}%
^{+}\cup\ell_{3}^{+}$ in a point $\mathbf{z}.$ If $\mathbf{z}$ lies on
$\ell_{i}^{+}$ then $\mathbf{y}$ is a convex combination of the two points
$\mathbf{y}_{i}$ and $\mathbf{z}.$ Hence, any $(f,g)$-divergence pair is a
convex combination of two points that are $2$-achievable. From the
construction in the proof of Theorem \ref{TheoremConvex} we see that any
$(f,g)$-divergence pair is 4-achievable.

An $f$-divergence on an arbitrary $\sigma$-algebra can be approximated by the
$f$-divergence on its finite sub-algebras. Any finite $\sigma$-algebra is a
Borel $\sigma$-algebra for a discrete space so for probability measures $P,Q$
on a $\sigma$-algebra the point $\left(  D_{f}\left(  P,Q\right)
,D_{g}\left(  P,Q\right)  \right)  $ is in the closure of 4-achievable points.
For any function pairs $(f,g)$ the intersection of the set of 2-achievable
points and the first quadrant is closed. 4-achievable points are convex
combinations of 2-achievable points so the intersection of the 4-achievable
points and the first quadrant is closed contains $\left(  D_{f}\left(
P,Q\right)  ,D_{g}\left(  P,Q\right)  \right)  $ even if $P,Q$ are measures on
a non-atomic $\sigma$-algebra.
\end{proof}

The set of $(f,g)$-achievable points that are 2-achievable can be parametrized
as $P=\left(  1-p,p\right)  $ and $Q=\left(  1-q,q\right)  .$ If we define
$\overline{\left(  1-p,p\right)  }=\left(  p,1-p\right)  $ then $D_{f}\left(
P,Q\right)  =D_{f}\left(  \overline{P},\overline{Q}\right)  .$ Hence we may
assume without loss of generality assume that $p\leq q$ and just have to
determine the image of the simplex $\Delta=\left\{  \left(  p,q\right)
\mid0\leq p\leq q\leq1\right\}  .$ This result makes it very easy to make a
numerical plot of the $(f,g)$-achievable point is 2-achievable and the joint
range is just the convex hull.

\input{simplex.TpX}

\section{Image of the triangle}

In order to determine the image of the triangle $\Delta$ we have to check what
happens at inner points and what happens at or near the boundary. Most inner
points are mapped into inner points of the range. On subsets of $\Delta$ where
the derivative matrix is non-singular the mapping $\left(  P,Q\right)
\rightarrow\left(  D_{f},D_{g}\right)  $ is open according to the open mapping
theorem from calculus. Hence, all inner points that are not mapped into
interior points of the range must satisfy
\[
\left\vert
\begin{array}
[c]{cc}%
\frac{\partial D_{f}}{\partial p} & \frac{\partial D_{g}}{\partial p}\\
\frac{\partial D_{f}}{\partial q} & \frac{\partial D_{g}}{\partial q}%
\end{array}
\right\vert =0.
\]
Depending on functions $f$ and $g$ this equation may be easy or difficult to
solve, but in most cases the solutions will lie on a 1-dimensional manifold
that will cut the triangle $\Delta$ into pieces, such that each piece is
mapped isomorphically into subsets of the range of $\left(  P,Q\right)
\rightarrow\left(  D_{f},D_{g}\right)  .$ Each pair of functions $(f,g)$ will
require its own analysis.

The diagonal $p=q$ in $\Delta$ is easy to analyze. It is mapped into $\left(
D_{f},D_{g}\right)  =\left(  0,0\right)  .$

\begin{lemma}
\label{uendelig}If $f\left(  0\right)  =\infty,$ and $\lim_{t\rightarrow0}%
\inf\frac{g\left(  t\right)  }{f\left(  t\right)  }=\beta_{0},$ then the
supremum of
\[
\beta\cdot D_{f}\left(  P,Q\right)  -D_{g}\left(  P,Q\right)
\]
over all distributions $P,Q$ is $\infty$ if $\beta>\beta_{0}.$

If $f^{\ast}\left(  0\right)  =\infty,$ and $\lim_{t\rightarrow\infty}
\inf\frac{g\left(  t\right)  }{f\left(  t\right)  }=\beta_{0},$ then the
supremum of
\[
\beta\cdot D_{f}\left(  P,Q\right)  -D_{g}\left(  P,Q\right)
\]
over all distributions $P,Q$ is $\infty$ if $\beta>\beta_{0}.$
\end{lemma}

If $g\left(  0\right)  =\infty,$ and $\lim_{t\rightarrow0}\sup\frac{g\left(
t\right)  }{f\left(  t\right)  }=\gamma_{0},$ then the supremum of
\[
D_{g}\left(  P,Q\right)  -\gamma D_{f}\left(  P,Q\right)
\]
over all distributions $P,Q$ is $\infty$ if $\gamma<\gamma_{0}.$

If $g^{\ast}\left(  0\right)  =\infty,$ and $\lim_{t\rightarrow\infty}
\sup\frac{g\left(  t\right)  }{f\left(  t\right)  }=\gamma_{0},$ then the
supremum of
\[
D_{g}\left(  Q,P\right)  -\gamma D_{f}\left(  Q,P\right)
\]
over all distributions $P,Q$ is $\infty$ if $\gamma<\gamma_{0}.$

\begin{proof}
Assume that
\[
f\left(  0\right)  =\infty\text{ \ and \ }\lim_{t\rightarrow0}\inf
\frac{g\left(  t\right)  }{f\left(  t\right)  }=\beta_{0}.
\]
The first condition implies
\[
D_{f}\left(  \left(  1,0\right)  ,\left(  1/2,1/2\right)  \right)  =\infty
\]
and the second condition implies that $g\left(  0\right)  =\infty$ and
\[
D_{g}\left(  \left(  1,0\right)  ,\left(  1/2,1/2\right)  \right)  =\infty.
\]
We have
\begin{multline*}
\frac{D_{g}\left(  \left(  p,1-p\right)  ,\left(  1/2,1/2\right)  \right)
}{D_{f}\left(  \left(  p,1-p\right)  ,\left(  1/2,1/2\right)  \right)  }\\
=\frac{g\left(  2p\right)  /2+g\left(  2\left(  1-p\right)  \right)
/2}{f\left(  2p\right)  /2+f\left(  2\left(  1-p\right)  \right)  /2}\\
=\frac{g\left(  2p\right)  +g\left(  2\left(  1-p\right)  \right)  }{f\left(
2p\right)  +f\left(  2\left(  1-p\right)  \right)  }.
\end{multline*}
Let $\left(  t_{n}\right)  _{n}$ be a sequence such that $\frac{g\left(
t_{n}\right)  }{f\left(  t_{n}\right)  }\rightarrow\beta$ for $n\rightarrow
\infty.$ Then
\[
\frac{D_{g}\left(  \left(  \frac{t_{n}}{2},1-\frac{t_{n}}{2}\right)  ,\left(
1/2,1/2\right)  \right)  }{D_{f}\left(  \left(  \frac{t_{n}}{2},1-\frac{t_{n}
}{2}\right)  ,\left(  1/2,1/2\right)  \right)  }\rightarrow\beta
\]
and the first result follows.

The other three cases follows by interchanging $f$ and $g,$ and/or replacing
$f$ by $f^{\ast}$ and $g$ by $g^{\ast}.$ We have used that
\[
\lim_{t\rightarrow0}\inf\frac{g^{\ast}\left(  t\right)  }{f^{\ast}\left(
t\right)  }=\lim_{t\rightarrow0}\inf\frac{tg\left(  t^{-1}\right)  }{tf\left(
t^{-1}\right)  }=\lim_{t\rightarrow\infty}\inf\frac{g\left(  t\right)
}{f\left(  t\right)  }.
\]

\end{proof}

\begin{proposition}
Assume that $f$ and $g$ are $C^{2}$ and that $f^{\prime\prime}\left(
1\right)  >0$ and $g^{\prime\prime}\left(  1\right)  >0.$ Assume that
$\lim_{t\rightarrow0}\inf\frac{g\left(  t\right)  }{f\left(  t\right)  }>0,$
and that $\lim_{t\rightarrow\infty}\inf\frac{g\left(  t\right)  }{f\left(
t\right)  }>0.$ Then there exists $\beta>0$ such that
\begin{equation}
D_{g}\left(  P,Q\right)  \geq\beta\cdot D_{f}\left(  P,Q\right)
\label{nederen}%
\end{equation}
for all distributions $P,Q.$
\end{proposition}

\begin{proof}
The inequality $\lim_{t\rightarrow0}\inf\frac{g\left(  t\right)  }{f\left(
t\right)  }>0$ implies that there exist $\beta_{0}$,$t_{0}>0$ such that
$g\left(  t\right)  \geq\beta_{0}f\left(  t\right)  $ for $t<t_{0}.$ The
Inequality $\lim_{t\rightarrow\infty}\inf\frac{g\left(  t\right)  }{f\left(
t\right)  }>0$ implies that there exists $\beta_{\infty}>0$ and $t_{\infty}>0$
such that $g\left(  t\right)  \geq\beta_{\infty}f\left(  t\right)  $ for
$t>t_{\infty.}$ According to Taylor's formula we have
\begin{align*}
f\left(  t\right)   &  =\frac{f^{\prime\prime}\left(  \theta\right)  }
{2}\left(  t-1\right)  ^{2},\\
g\left(  t\right)   &  =\frac{g^{\prime\prime}\left(  \eta\right)  }{2}\left(
t-1\right)  ^{2}%
\end{align*}
for some $\theta$ and $\eta$ between $1$ and $t.$ Hence
\[
\frac{g\left(  t\right)  }{f\left(  t\right)  }=\frac{f^{\prime\prime}\left(
\theta\right)  }{g^{\prime\prime}\left(  \eta\right)  }\rightarrow
\frac{f^{\prime\prime}\left(  1\right)  }{g^{\prime\prime}\left(  1\right)
}\text{ for }t\rightarrow1.
\]
Therefore there there exists $\beta_{1}>0$ and an interval $\left]
t_{-},t_{+}\right[  $ around $1$ such that $\frac{g\left(  t\right)
}{f\left(  t\right)  }\geq\beta_{1}$ for $t\in\left]  t_{-},t_{+}\right[  .$
The function $t\rightarrow\frac{g\left(  t\right)  }{f\left(  t\right)  }$ is
continuous on the compact set $\left[  t_{0},t_{-}\right]  \cup\left[
t_{+},t_{\infty}\right]  $ so it has a minimum $\tilde{\beta}>0$ on this set.
Inequality \ref{nederen} holds for $\beta=\min\left\{  \beta_{0},\beta
_{1},\beta_{\infty},\tilde{\beta}\right\}  .$
\end{proof}

\section{Examples}

In this section we shall see a number of examples of how the method developed
i this paper can be applied to determine the joint range for some pairs of
$f$-divergences. Some of these results are known and others are new. We will
not spell out all the details but shall restrict to the main flow of the
argument that will lead to the joint range.

\subsection{Power divergence of order 2 and 3}

We have
\begin{align*}
f\left(  t\right)   &  =\phi_{2}(t){,}\\
g\left(  t\right)   &  =\phi_{3}(t){.}%
\end{align*}

In this case we have
\begin{gather*}
D_{f}\left(  \left(  p,1-p\right)  ,\left(  q,1-q\right)  \right)
=\ \ \ \ \ \ \ \ \ \ \ \ \ \ \ \ \ \ \ \ \ \ \ \\
\frac{1}{2}\left(  \frac{\left(  p-q\right)  ^{2}}{q}+\frac{\left(
p-q\right)  ^{2}}{1-q}\right)  ,\\
D_{g}\left(  \left(  p,1-p\right)  ,\left(  q,1-q\right)  \right)
=\ \ \ \ \ \ \ \ \ \ \ \ \ \ \ \ \ \ \ \ \ \ \ \\
\ \ \ \ \ \ \ \ \ \ \ \ \ \ \ \ \frac{1}{6}\left(  \left(  \frac{p}{q}\right)
^{3}q+\left(  \frac{1-p}{1-q}\right)  ^{3}\left(  1-q\right)  -1\right)  .
\end{gather*}
First we determine the image of the triangle. The derivatives are
\begin{align*}
\frac{\partial D_{f}}{\partial p}  &  =\frac{2}{2}\cdot\frac{\left(
p-q\right)  }{\left(  1-q\right)  q}~,\\
\frac{\partial D_{f}}{\partial q}  &  =\frac{1}{2}\cdot\frac{\left(
2pq-q-p\right)  \left(  p-q\right)  }{\left(  1-q\right)  ^{2}q^{2}}~,\\
\frac{\partial D_{g}}{\partial p}  &  =\frac{-3}{6}\cdot\frac{\left(
2pq-q-p\right)  \left(  p-q\right)  }{\left(  1-q\right)  ^{2}q^{2}}~,\\
\frac{\partial D_{g}}{\partial q}  &  =\frac{2}{6}\cdot\frac{\left(
\begin{array}
[c]{c}%
pq+p^{2}+q^{2}-\\
3pq^{2}-3p^{2}q+3p^{2}q^{2}%
\end{array}
\right)  \allowbreak\left(  p-q\right)  }{\left(  q-1\right)  ^{3}q^{3}}~.
\end{align*}
The determinant of derivatives is
\begin{multline*}
\left\vert
\begin{array}
[c]{cc}%
\frac{\partial D_{f}}{\partial p} & \frac{\partial D_{g}}{\partial p}\\
\frac{\partial D_{f}}{\partial q} & \frac{\partial D_{g}}{\partial q}%
\end{array}
\right\vert =\\
\frac{\left(  p-q\right)  ^{2}}{12q^{4}\left(  1-q\right)  ^{4}}\left\vert
\begin{array}
[c]{cc}%
2 & 3p+3q-6pq\allowbreak\\
2pq-q-p & \left(
\begin{array}
[c]{c}%
6pq^{2}-2p^{2}-2q^{2}\\
-2pq+6p^{2}q-6p^{2}q^{2}\allowbreak
\end{array}
\right)
\end{array}
\right\vert \\
=-\frac{1}{12}\left(  \frac{p-q}{q\left(  1-q\right)  }\right)  ^{4}.
\end{multline*}

$\allowbreak$We see that the determinant of derivatives is different from zero
for $p\neq q$ so the interior of $\Delta$ is mapped one-to-one to the image.
Hence we just have to determine the image of points on the boundary of
$\Delta$ (or near the boundary if undefined on the boundary).

For $P=\left(  1,0\right)  $ and $Q=\left(  1-q,q\right)  $ we get
\begin{align*}
D_{f}\left(  P,Q\right)   &  =\frac{1}{2}\left(  q+\frac{q^{2}}{1-q}\right)
=\frac{1}{2}\left(  \frac{1}{1-q}-1\right)  ,\\
D_{g}\left(  P,Q\right)   &  =\frac{1}{6}\left(  \frac{1}{\left(  1-q\right)
^{2}}-1\right)  =\frac{1}{6}\frac{\left(  2-q\right)  q}{\left(  1-q\right)
^{2}}.
\end{align*}
The first equation leads to
\[
q=\left(  1-\frac{1}{2D_{f}+1}\right)
\]
and hence
\[
D_{g}=\frac{2}{3}D_{f}\left(  D_{f}+1\right)  .
\]
We have
\[
\frac{f\left(  t\right)  }{g\left(  t\right)  }=\frac{{\frac{t^{2}
-2(t-1)-1}{2}}}{{\frac{t^{3}-3(t-1)-1}{6}}}\rightarrow\infty\text{ for
}t\rightarrow\infty.
\]
All points $\left(  0,s\right)  ,s\in\left[  0,\infty\right)  $ are in the
closure of the range of $\left(  P,Q\right)  \rightarrow\left(  D_{f}
,D_{g}\right)  .$ By combing these two results we see that the range consists
of the point $\left(  0,0\right)  ,$ all points on the curve $\left(
x,\frac{2}{3}x\left(  x+1\right)  \right)  ,x\in\left(  0,\infty\right)  $,
and all point above this curve.

Similar results holds for any pair of power divergences, but for other pairs
than $\left(  D_{2},D_{3}\right)  $ the computations become much more involved.

Note that the R\'{e}nyi divergences are monotone functions of the power
divergences so our results easily translate into the results on R\'{e}nyi
divergences. More details on R\'{e}nyi divergences can be found in
\cite{Erven2010}.

\subsection{Total variation and $\chi^{2}$-divergence}

In this case we have
\begin{align*}
f\left(  x\right)   &  =\left\vert x-1\right\vert ,\\
g\left(  x\right)   &  =\frac{1}{2}\left(  x-1\right)  ^{2}.
\end{align*}
The function $f$ is not differentiable but on the triangle $\Delta$ we have
$p\leq q$ and
\begin{align*}
D_{f}\left(  P,Q\right)   &  =q\left\vert \frac{p}{q}-1\right\vert +\left(
1-q\right)  \left\vert \frac{1-p}{1-q}-1\right\vert \\
&  =2\left(  q-p\right)  .
\end{align*}
Hence $D_{f}\left(  P,Q\right)  $ is $C^{\infty}$ on $\Delta$ although $f$ is
not differentiable. We get
\begin{align*}
\frac{\partial D_{f}}{\partial p} &  =-2~,\\
\frac{\partial D_{f}}{\partial q} &  =2~,\\
\frac{\partial D_{g}}{\partial p} &  =\frac{\left(  p-q\right)  }{\left(
1-q\right)  q}~,\\
\frac{\partial D_{g}}{\partial q} &  =\frac{\left(  2pq-q-p\right)  \left(
p-q\right)  }{2\left(  1-q\right)  ^{2}q^{2}}~.
\end{align*}
Hence%
\begin{align*}
\left\vert
\begin{array}
[c]{cc}%
\frac{\partial D_{f}}{\partial p} & \frac{\partial D_{g}}{\partial p}\\
\frac{\partial D_{f}}{\partial q} & \frac{\partial D_{g}}{\partial q}%
\end{array}
\right\vert  &  =\left\vert
\begin{array}
[c]{cc}%
-2 & 2\\
\frac{\left(  p-q\right)  }{\left(  1-q\right)  q} & \frac{\left(
2pq-q-p\right)  \left(  p-q\right)  }{2\left(  1-q\right)  ^{2}q^{2}}%
\end{array}
\right\vert \\
&  =-2\frac{\left(  q-p\right)  ^{2}\left(  q-1/2\right)  }{\left(
1-q\right)  ^{2}q^{2}}.
\end{align*}
$\allowbreak$The mapping $\Delta$ to the range of $\left(  D_{f},D_{g}\right)
$ is singular for $q=1/2.$ The line $p\rightarrow\left(  p,1/2\right)  $ is
mapped into the curve%
\begin{align*}
p &  \rightarrow\left(  D_{f}\left(  P,Q\right)  ,D_{g}\left(  P,Q\right)
\right)  \\
&  =\left(  2\left(  p-\frac{1}{2}\right)  ,2\left(  p-1/2\right)
^{2}\right)  .
\end{align*}
If the total variation is denoted $V$ this curve satisfies $\chi^{2}=\frac
{1}{2}V^{2}$ and points satisfying $\chi^{2}\geq\frac{1}{2}V^{2}$ are
2-achievable. The inequality $\chi^{2}\geq\frac{1}{2}V^{2}$ has been proved
previously by a different method \cite{Gibbs2002}.

\subsection{Total variation and LeCam divergence}

On the triangle $\Delta$ we have
\begin{align*}
D_{f}\left(  P,Q\right)   &  =2\left(  q-p\right)  ,\\
D_{g}\left(  P,Q\right)   &  =\frac{1}{4}\left(  \frac{\left(  p-q\right)
^{2}}{p+q}+\frac{\left(  p-q\right)  ^{2}}{2-p-q}\right)  .
\end{align*}
The derivatives of the LeCam divergence is
\begin{align*}
\frac{\partial}{\partial p}D_{g}\left(  P,Q\right)   &  =\frac{\left(
p-q\right)  \left(  p+3q-2pq-2q^{2}\right)  }{\left(  p+q\right)  ^{2}\left(
2-p-q\right)  ^{2}},\\
\frac{\partial}{\partial q}D_{g}\left(  P,Q\right)   &  =\frac{\allowbreak
\left(  2pq-q-3p+2p^{2}\right)  \allowbreak\left(  p-q\right)  }{\left(
p+q\right)  ^{2}\left(  p+q-2\right)  ^{2}}.
\end{align*}
Hence%
\begin{multline*}
\left\vert
\begin{array}
[c]{cc}%
\frac{\partial D_{f}}{\partial p} & \frac{\partial D_{g}}{\partial p}\\
\frac{\partial D_{f}}{\partial q} & \frac{\partial D_{g}}{\partial q}%
\end{array}
\right\vert \\
=\left\vert
\begin{array}
[c]{cc}%
-2 & 2\\
\frac{\left(  p-q\right)  \left(  p+3q-2pq-2q^{2}\right)  }{\left(
p+q\right)  ^{2}\left(  2-p-q\right)  ^{2}} & \frac{\left(  2pq-q-3p+2p^{2}%
\right)  \allowbreak\left(  p-q\right)  }{\left(  p+q\right)  ^{2}\left(
p+q-2\right)  ^{2}}%
\end{array}
\right\vert \\
=\frac{4\left(  1-p-q\right)  \left(  q-p\right)  ^{2}}{\left(  p+q\right)
^{2}\left(  p+q-2\right)  ^{2}}.
\end{multline*}
The mapping is singular for $q=1-p.$ We get the curve%
\begin{align*}
p &  \rightarrow\left(  2\left(  p-\left(  1-p\right)  \right)  ,\frac{\left(
p-\left(  1-p\right)  \right)  ^{2}}{p+\left(  1-p\right)  }+\frac{\left(
p-\left(  1-p\right)  \right)  ^{2}}{2-p-\left(  1-p\right)  }\right)  \\
&  =\left(  4\left(  p-\frac{1}{2}\right)  ,2\left(  p-\frac{1}{2}\right)
^{2}\right)  .
\end{align*}
If total variation is denoted $V$ then the curve is $D_{g}=\frac{1}{8}V^{2}$
and any point above this curve is achievable.

\subsection{Information divergence and reversed information divergence}

In this case we have%
\begin{align*}
f\left(  t\right)   &  =t\ln t,\\
g\left(  t\right)   &  =-\ln t.
\end{align*}
We see that $g\left(  0\right)  =\infty$ and that $\frac{g\left(  t\right)
}{f\left(  t\right)  }\rightarrow\infty$ for $t\rightarrow0.$ Lemma
\ref{uendelig} implies that the supremum of
\[
D_{g}\left(  P,Q\right)  -\gamma D_{f}\left(  P,Q\right)  =D\left(  Q\Vert
P\right)  -\gamma D\left(  P\Vert Q\right)
\]
over all distributions $P,Q$ is $\infty$ for any $\gamma<\infty.$ Similarly
the supremum of
\[
D\left(  P\Vert Q\right)  -\gamma D\left(  Q\Vert P\right)
\]
over all distributions $P,Q$ is $\infty$ for any $\gamma<\infty.$ Since
$\left(  0,0\right)  $ is in the range and the range is convex, the range
consist of all interior points of the first quadrant and the point $\left(
0,0\right)  .$

\section{Acknowledgement}

The authors thank Job Bri\"{e}t and Tim van Erven for comments to a draft of
this paper.

This work was supported by the European Network of Excellence and the GA\v{C}R
grants 102/07/1131 and 202/10/0618.

\setlength{\itemsep}{5pt}

\bibliographystyle{ieeetr}
\bibliography{database1}

\end{document}